\documentclass[reqno]{amsart}

\usepackage{amssymb,latexsym,amsfonts,amsmath,bbm}
\usepackage{mathrsfs}
\usepackage{graphicx}
\usepackage{paralist}
\usepackage{eso-pic}
\usepackage{epsfig}
\usepackage{mathtools}
\usepackage{epstopdf}
\usepackage{lipsum}
\usepackage{cite}
\usepackage{subfig}

\usepackage{dsfont}
\newcommand{\R}{{\mathds{R}}}
\newcommand{\Rp}{{\mathds{R}_{> 0}}}
\newcommand{\Rpz}{{\mathds{R}_{\geq 0}}}
\newcommand{\N}{{\mathds{N}}}
\newcommand{\Np}{{\mathds{N}_{\geq 1}}}
\newcommand{\I}{{\mathds{I}}}

\usepackage{optidef}

\DeclareRobustCommand{\legendsquare}[1]{%
	\textcolor{#1}{\rule{1.5ex}{1.5ex}}%
}

\DeclareRobustCommand\sampleline[1]{%
	\tikz\draw[#1] (0,0) (0,\the\dimexpr\fontdimen22\textfont2\relax)
	-- (1.5em,\the\dimexpr\fontdimen22\textfont2\relax);%
}

\definecolor{x1}{rgb}{0 0.5 0.2}
\definecolor{x2}{rgb}{0 0 1}
\definecolor{xh1}{rgb}{0.9 0.5 0}
\definecolor{xh2}{rgb}{1 0.1 0.4}

\definecolor{X1}{rgb}{0 0 1}
\definecolor{X2}{rgb}{1 0 0}

\DeclareRobustCommand{\legendsquare}[1]{%
	\textcolor{#1}{\rule{1.5ex}{1.5ex}}%
}

\definecolor{unsafe}{rgb}{0.9, 0, 0.1}
\definecolor{initt}{rgb}{0.1, 0.3, 1}
\definecolor{state}{rgb}{0, 0.9, .7}

\definecolor{START}{rgb}{0 0.7 0.6}
\definecolor{TARGET}{rgb}{0.5 0.8 1}
\definecolor{OBSTACLES}{rgb}{0.9 0 0.3}

\usepackage{tikz}
\usetikzlibrary{calc,shapes,arrows}
\definecolor{denim}{rgb}{0.08, 0.38, 0.74}
\definecolor{deeppink}{rgb}{1.0, 0.08, 0.58}

\usetikzlibrary{positioning}
\usetikzlibrary{
	arrows.meta,                        
	backgrounds,                        
	ext.paths.ortho,                    
	ext.positioning-plus,               
	ext.node-families.shapes.geometric, 
	calc}                               
\tikzset{
	basic box/.style={
		shape=rectangle, rounded corners, align=center, draw=#1, fill=#1!25, minimum height=20mm, minimum width = 20mm},
	header node/.style={
		node family/width=header nodes,
		font=\strut\Large\ttfamily,
		text depth=+.3ex, fill=white, draw},
	header/.style={%
		inner ysep=+1.5em,
		append after command={
			\pgfextra{\let\TikZlastnode\tikzlastnode}
			node [header node] (header-\TikZlastnode) at (\TikZlastnode.north) {#1}
			node [span=(\TikZlastnode)(header-\TikZlastnode)]
			at (fit bounding box) (h-\TikZlastnode) {}
		}
	},
	fat blue line/.style={ultra thick, blue}
}
\usetikzlibrary{matrix}

\definecolor{mediumblue}{rgb}{0.0, 0.0, 0.8}
\definecolor{mediumcandyapplered}{rgb}{0.89, 0.02, 0.17}
\definecolor{nazar}{rgb}{0.7, 0.5, 0.9}
\makeatletter
\let\NAT@parse\undefined
\makeatother
\usepackage[colorlinks=true, citecolor=black, linkcolor=mediumblue, urlcolor = nazar, final]{hyperref}

\definecolor{lightblue}{rgb}{0.30,0.75,0.93}

\DeclareRobustCommand\sampleline[1]{%
	\tikz\draw[#1] (0,0) (0,\the\dimexpr\fontdimen22\textfont2\relax)
	-- (2em,\the\dimexpr\fontdimen22\textfont2\relax);%
}

\usepackage{algorithm}
\usepackage{algorithmic}

\usepackage{pgfplots}

\usepackage{xspace}

\newtheorem{theorem}{Theorem}[section]
\newtheorem{lemma}[theorem]{Lemma}

\newtheorem{problem}[theorem]{Problem}

\newtheorem{definition}[theorem]{Definition}

\newtheorem{remark}[theorem]{Remark}

\numberwithin{equation}{section}

\usepackage{fancyhdr}

\newenvironment{nouppercase}{%
	\renewcommand{\uppercasenonmath}[1]{}}{}

\usepackage{amsmath}
\usepackage[many]{tcolorbox}
\usetikzlibrary{calc}
\tcbuselibrary{skins}

\newtcolorbox{resp}[1][]{%
	enhanced jigsaw,%
	colback=gray!5!white,%
	colframe=gray!80!black,%
	size=small,%
	boxrule=1pt,%
	halign title=flush center,%
	coltitle=black,%
	breakable,%
	drop shadow=black!50!white,%
	attach boxed title to top left={xshift=1cm,yshift=-\tcboxedtitleheight/2,yshifttext=-\tcboxedtitleheight/2},%
	minipage boxed title=3cm,%
	boxed title style={%
		colback=white,%
		size=fbox,%
		boxrule=1pt,%
		boxsep=2pt,%
		underlay={%
			\coordinate (dotA) at ($(interior.west) + (-0.5pt,0)$);
			\coordinate (dotB) at ($(interior.east) + (0.5pt,0)$);
			\begin{scope}[gray!80!black]
				\fill (dotA) circle (2pt);
				\fill (dotB) circle (2pt);
			\end{scope}
		}%
	},%
	#1%
}

\begin{document}

\begin{abstract}
	Finite abstractions (a.k.a. symbolic models) offer an effective scheme for approximating the complex continuous-space systems with simpler models in the \emph{discrete-space} domain. A crucial aspect, however, is to establish a formal relation between the original system and its symbolic model, ensuring that a discrete controller designed for the symbolic model can be effectively implemented as a \emph{hybrid} controller (using an interface map) for the original system. This task becomes even more challenging when the exact mathematical model of the continuous-space system is unknown. To address this, the existing literature mainly employs \emph{scenario-based} data-driven methods, which require collecting a large amount of data from the original system. In this work, we propose a data-driven framework that utilizes only \emph{two input-state trajectories} collected from unknown nonlinear polynomial systems to synthesize a hybrid controller, enabling the desired behavior on the unknown system through the controller derived from its symbolic model. To accomplish this, we employ the concept of alternating simulation functions (ASFs) to quantify the closeness between the state trajectories of the unknown system and its data-driven symbolic model. By satisfying a specific rank condition on the collected data, which intuitively ensures that the unknown system is persistently excited, we directly design an ASF and its corresponding hybrid controller using finite-length data without explicitly identifying the unknown system, while providing correctness guarantees. This is achieved through proposing a data-based sum-of-squares (SOS) optimization program, enabling a systematic approach to the design process. We illustrate the effectiveness of our data-driven approach through a case study.
\end{abstract}

\title{{\LARGE Abstraction-based Control of Unknown Continuous-Space Models  \vspace{0.2cm}\\ with Just Two Trajectories \vspace{0.4cm}}}

\author{{\bf {\large Behrad Samari}}}
\author{{\bf {\large Mahdieh Zaker}}}
\author{{{\bf {\large Abolfazl Lavaei}}}\vspace{0.4cm}\\
	{\normalfont School of Computing, Newcastle University, United Kingdom}\\
\texttt{\{behrad.samari, mahdieh.zaker, abolfazl.lavaei\}@newcastle.ac.uk}}

\pagestyle{fancy}
\lhead{}
\rhead{}
  \fancyhead[OL]{Behrad Samari, Mahdieh Zaker, and Abolfazl Lavaei}

  \fancyhead[EL]{Abstraction-based Control of Unknown Continuous-Space Models with Just Two Trajectories}
  \rhead{\thepage}
 \cfoot{}
 
\begin{nouppercase}
	\maketitle
\end{nouppercase}

\section{Introduction}\label{Intro}
Over the past two decades, formal methods have garnered significant attention as a promising technique for synthesizing controllers that enforce high-level logic specifications, such as those expressed by linear temporal logic (LTL) formulae~\cite{pnueli1977temporal}. These methods enable the reliable control of complex dynamical systems while ensuring rigorous mathematical guarantees. However, synthesizing formal controllers for complex dynamical systems, especially those with nonlinear terms, to satisfy high-level logic properties remains inherently challenging, primarily due to the computational complexity associated with \emph{continuous} state and input spaces.

In order to address the encountered computational challenges, a promising approach in the literature is to construct finite abstractions (also referred to as symbolic models), which serve as approximate representations of continuous-space systems where each finite state corresponds to a group of continuous states from the original (concrete) systems. These finite abstractions can then be used as substitutes for the original systems, enabling verification or controller synthesis on the abstract models, and subsequently, the results can be translated back to the original complex systems. This is achieved by establishing a similarity relationship between the state trajectories of the original system and its finite abstraction using alternating simulation functions (ASFs), ensuring that the original system satisfies the same property as the abstract model.

Symbolic models are generally divided into two distinct classes: \emph{sound and complete} abstractions~\cite{tabuada2009verification}. Complete abstractions provide both \emph{necessary and sufficient} guarantees, implying that a controller enforcing a desired property on the abstraction exists \emph{if and only if} such a controller exists for the original system. On the contrary, sound abstractions provide only \emph{sufficient} guarantees, meaning the failure to synthesize a controller using a sound abstraction does not necessarily imply the absence of a controller for the original system. Despite the fruitfulness of symbolic model methods in complex systems analysis, the precise mathematical model of the original system is required for the construction of such (sound or complete) abstract models (see \emph{e.g.,} \cite{tabuada2009verification,zamani2014symbolic,girard2009hierarchical,reissig2016feedback}). Nonetheless, in many practical scenarios, the mathematical models are either unknown or challenging to obtain.

\noindent{\textbf{Data-Driven Techniques.}} To address this key difficulty, data-driven techniques have been established in two main categories: \emph{indirect and direct}~\cite{dorfler2022bridging}. Indirect data-driven approaches focus on learning unknown dynamical models via system identification techniques~\cite{hou2013model}. However, accurately identifying these models is often computationally demanding, especially for unknown complex systems with nonlinear dynamics. Furthermore, even after identifying an approximate model, it is necessary to derive the similarity relation between the identified dynamics and its symbolic abstraction using model-based techniques. Consequently, the inherent complexity arises at two levels: model identification and deriving the required relation. In contrast, direct data-driven approaches evade the system identification process, instead leveraging data \emph{directly} to construct symbolic models and establish similarity relations~\cite{dorfler2022bridging}.

Data-driven techniques (including indirect and direct) have been widely used thus far to construct symbolic models for unknown systems. Existing results include building symbolic models for unknown systems with a guaranteed confidence~\cite{lavaei2022data}, construction of symbolic models for monotone systems with disturbances~\cite{makdesi2021efficient}, sample-based sequential~\cite{banse2023data1} and learning-based abstract construction~\cite{hashimoto2022learning}, probabilistic verification of logic specifications for deterministic systems~\cite{coppola2022data}, and synthesizing controllers for unknown systems through symbolic models with provable guarantees~\cite{devonport2021symbolic,coppola2024data,banse2023data2,ajeleye2023data}. In addition, \cite{lavaei2023symbolic} and \cite{samari2024data} propose abstraction-based controller synthesis frameworks for interconnected networks, offering probabilistic and deterministic (\emph{i.e.,} confidence 1) correctness guarantees, respectively.

The majority of the aforementioned literature focuses on symbolic model construction and ensuring their correctness using \emph{scenario-based} approaches~\cite{calafiore2006scenario,campi2009scenario}. While these methods have proven effective for symbolic model construction in unknown systems, they rely on the assumption that the data are \emph{independent and identically distributed (i.i.d.)}. Consequently, only a single input-output (or input-state) data pair can be extracted from each trajectory \cite{calafiore2006scenario}, necessitating \emph{multiple independent trajectories}—potentially up to millions in practical scenarios—due to the \emph{exponential sample complexity} inherent in scenario-based techniques when applied to the robust problem~\cite{esfahani2014performance}. As a result, scenario-based methods are particularly well-suited for system simulators, where generating numerous independent trajectories is practical.

In contrast, \emph{trajectory-based} approaches~\cite{de2019formulas} do not require i.i.d. data, allowing a single trajectory to be collected over a specified horizon for analysis, provided that a rank condition is satisfied, ensuring the data is sufficiently informative for the analysis. Several studies have adopted a trajectory-based approach over the past few years, including, but not limited to, synthesizing robust model predictive controllers for linear time-invariant systems~\cite{berberich2020data}, input-to-state stabilizing controllers~\cite{chen2024data}, and safety controllers utilizing control barrier certificates for continuous-time~\cite{nejati2022data,akbarzadeh2024data} and discrete-time nonlinear polynomial systems~\cite{samari2024single}, and for discrete-time linear control systems with
wireless communication networks~\cite{akbarzadeh2024data1}. Nevertheless, to the best of our knowledge, no existing research has explored finite abstractions using a trajectory-based approach, due to the inherent complexity of establishing similarity relations.

\noindent{\textbf{Original Contributions.}} Driven by the challenges outlined in prior data-driven studies on constructing symbolic models, this work aims to develop a data-driven trajectory-based methodology for unknown discrete-time input-affine nonlinear systems with polynomial dynamics. Our scheme focuses on constructing symbolic models and establishing the similarity relation between the original, unknown systems and their symbolic counterparts. The cornerstone of our framework is the employment of alternating simulation functions (ASFs) with the pursuit of capturing the mismatch between the original system and its symbolic model. Within the proposed approach, we start with gathering just two trajectories (\emph{i.e.,} two sets of input-state data) from the unknown system while ensuring a specific rank condition is satisfied (\emph{i.e.,} a required minimum number of data is gathered), which means that the collected data is rich enough for conducting the analysis. We then construct ASFs directly from the finite-length samples while ensuring correctness guarantees. This is then followed by synthesizing a controller for the symbolic model and transferring the result back to the original unknown system, leveraging a hybrid interface function. We showcase the efficacy of our approach via a case study.
\section{Problem Formulation}\label{Prob Form}

\subsection{Notation}
The set of real numbers is denoted by $\R$, with $\Rpz$  and $\Rp$ representing the sets of non-negative and positive real numbers, respectively. The set of non-negative integers is expressed as $\N$, while $\Np$ signifies the set of positive integers. 
The $n \times n$ identity matrix is denoted by $\I_n$, while $\mathbf{0}_n$ refers to the zero vector of dimension $n$. For any matrix $A$, its transpose is represented by $A^\top$. The symbol $\star$ indicates the transposed element in the symmetric position of a symmetric matrix.
The notation $P \succ 0$ ($P \succeq 0$) indicates that the \emph{symmetric} matrix $P$ is positive definite (positive semi-definite). For a square matrix $P$, the smallest and largest eigenvalues are denoted by $\lambda_{\min}(P)$ and $\lambda_{\max}(P)$, respectively.
The Euclidean norm of a vector $x \in \R^{n}$ is denoted by $\Vert x \Vert$, while for a given matrix $A$, $\Vert A \Vert$ represents its induced 2-norm. The horizontal concatenation of vectors $x_i \in \R^n$ into an $n \times N$ matrix is written as $\begin{bmatrix} x_1 & \hspace{-0.2cm} x_2 & \hspace{-0.2cm} \dots & \hspace{-0.2cm} x_N \end{bmatrix}$. The union among the sets $\mathsf{X}_i, i\in\{1,\dots, n_x\}$ is represented as $\cup_i \mathsf{X}_i$. Given sets $X$ and $Y$, a relation $\mathscr{R} \subseteq X \times Y$ is a subset of the Cartesian product $X \times Y$ that relates $x \in X$ to $y \in Y$ if $(x, y) \in \mathscr{R}$, equivalently denoted by $x \mathscr{R} y$.

\subsection{Discrete-Time Control Systems}
We commence by describing the discrete-time control systems for which we desire to construct a symbolic model.

\begin{definition}[dt-IANSP]\label{def: sys}A discrete-time input-affine nonlinear system with polynomial dynamics (dt-IANSP) is described as  
	\begin{align}\label{eq: dt-IANSP}
		\Sigma\!: x^+ = A\mathcal{M}(x) + Bu,
	\end{align}  
	where \(x^+\) denotes the state variables at the subsequent time step, specifically defined as \(x^+ \coloneq x(k + 1), \; k \in \N\). The matrix \(A \in \R^{n \times M}\) represents the system matrix, \(B \in \R^{n \times m}\) is the control matrix, and \(\mathcal{M}(x) \in \R^M\), with \(\mathcal{M}(\mathbf{0}_n) = \mathbf{0}_M\), is a vector comprising monomials of the state variables \(x \in X\),  with \(X \subseteq \R^n\) denoting the state set. The control input \(u \in U\) is associated with the system, where \(U \subseteq \R^m\) represents the control input set.  
	The dt-IANSP in \eqref{eq: dt-IANSP} is represented by the tuple \(\Sigma = (A, \mathcal{M}, B, X, U)\).
\end{definition}

Both matrices \(A\) and \(B\) are considered \emph{unknown} in our work, aligning with real-world scenarios. While the exact form of the monomial $\mathcal{M}(x)$ is assumed to be unknown, we are provided with either an \emph{upper bound on its maximum degree}—enabling the construction of $\mathcal{M}(x)$ to encompass all possible combinations of states up to that degree—or an extended vector of monomials, referred to as the \emph{extended dictionary}, which includes all monomials relevant to the system dynamics, along with additional terms that may be unnecessary, to ensure thorough representation of unknown systems. We refer to dt-IANSP \eqref{eq: dt-IANSP} as a system with a \emph{(partially) unknown} model.

\begin{remark}[On $\mathcal{M}(x)$]\label{rem: M}Note that knowing either an upper bound or an extended dictionary for $\mathcal{M}(x)$ is not a restrictive assumption. In fact, in numerous practical applications, such as electrical and mechanical systems, the system dynamics are often derived from first principles. However, the precise parameters of the system, namely the matrices \(A\) and \(B\), typically remain unknown, as is the case in our work.
\end{remark}

Having defined the dt-IANSP in \eqref{eq: dt-IANSP}, we now move on to present how symbolic models can be constructed in the following subsection.

\subsection{Symbolic Models}\label{Sym}
Analyzing the dt-IANSP $\Sigma$, which evolves in a \emph{continuous-space domain}, presents substantial challenges. To address this, one can approximate the dt-IANSP \(\Sigma\) with a symbolic model defined by \emph{discrete} sets of states and inputs \cite{pola2016symbolic}. The approximation process begins by partitioning the continuous state and input sets into finite segments, denoted as \(X = \cup_i \mathsf{X}_i\) and \(U = \cup_i \mathsf{U}_i\), respectively. Representative points \(\hat{x}_i \in \mathsf{X}_i\) and \(\hat{u}_i \in \mathsf{U}_i\) are then selected to serve as \emph{discrete} state variables and control inputs. The following definition formally outlines the required procedure for constructing a symbolic model \cite{pola2016symbolic,swikir2019compositional}.

\begin{definition}[Symbolic Model Construction]\label{def: symbolic}Given a continuous-space dt-IANSP \(\Sigma = (A, \mathcal{M}, B,\\ X, U)\), the corresponding symbolic model \(\hat{\Sigma}\) is described as  
	\begin{align}\label{eq: symbolic}
		\hat{\Sigma} = (\hat f, \hat{X}, \hat{U}),
	\end{align}  
	where the discrete state and input sets are constructed as \(\hat{X} \coloneq \big\{\hat{x}_i , i = 1, \ldots, n_{\hat{x}}\big\}\) and \(\hat{U} \coloneq \big\{\hat{u}_i , i = 1, \ldots, n_{\hat{u}}\big\}\), respectively. The transition map $\hat f: \hat X \times \hat U \to \hat X$ is defined as  
	\begin{align}\label{symbolic dyn}
		\hat{\Sigma}\!: \hat x^+ = \hat f(\hat x, \hat u) = \Pi(A\mathcal{M}(\hat{x}) + B\hat{u}),
	\end{align}  
	where the quantization map \(\Pi: X \to \hat{X}\) assigns representative points \(\hat{x} \in \hat{X}\) and \(\hat{u} \in \hat{U}\) to any \(x \in X\) and \(u \in U\) from the corresponding partitions. This quantization map satisfies the inequality  
	\begin{align}\label{eq: quant}
		\|\Pi(x) - x\| \leq \delta, \quad \forall x \in X,
	\end{align}  
	where \(\delta \coloneq \sup \big\{\|x - x'\| , x, x' \in \mathsf{X}_i, \; i = 1, 2, \ldots, n_{\hat{x}}\big\}\) is the state discretization parameter.
\end{definition}

Since the symbolic models defined in \eqref{eq: symbolic} operate within a discrete domain, algorithmic methods from computer science can be leveraged to synthesize controllers that enforce complex logical properties \cite{baier2008principles}, \emph{e.g.,} reach-while-avoid. The key challenge, however, lies in ensuring that the properties of interest illustrated by the symbolic models are accurately transferred to the original systems. To address this, a similarity relation between the state trajectories of the two systems must be established, which is achieved using the concept of \emph{alternating simulation functions}, as detailed in the following subsection.

\subsection{Alternating Simulation Functions}
We begin by presenting the concept of an alternating simulation function between a dt-IANSP and its symbolic model. The following definition formalizes this concept \cite{tabuada2009verification,swikir2019compositional}.

\begin{definition}[ASF]\label{def: ASF}Consider a dt-IANSP \(\Sigma = (A, \mathcal{M}, B, X, U)\) and its symbolic model \(\hat{\Sigma} = (\hat{f}, \hat{X}, \hat{U})\). A function \(\pmb{\mathcal{S}}: X \times \hat{X} \to \Rpz\) is called an alternating simulation function (ASF) from \(\hat{\Sigma}\) to \(\Sigma\), denoted as \(\hat{\Sigma} \preceq_{\pmb{\mathcal{S}}} \! \Sigma\), if the following conditions are satisfied:
	\begin{subequations}
		\begin{itemize}
			\item For all \(x \in X\) and \(\hat{x} \in \hat{X}\):
			\begin{align}
				\alpha \|x - \hat{x}\|^2 \leq \pmb{\mathcal{S}}(x, \hat{x}), \label{eq: con1-def}
			\end{align}
			\item For all \(x \in X\), \(\hat{x} \in \hat{X}\), and \(\hat{u} \in \hat{U}\), there exists \(u \in U\) such that:
			\begin{align}
				\pmb{\mathcal{S}}\big(A\mathcal{M}(x) + Bu, \Pi(A\mathcal{M}(\hat{x}) + B\hat{u})\big) \leq \gamma \pmb{\mathcal{S}}(x, \hat{x}) + \rho \|\hat{u}\| + \psi, \label{eq: con2-def}
			\end{align}
		\end{itemize}
	\end{subequations}
	for some \(\alpha, \psi \in \Rp\), \(\rho \in \Rpz\), and \(\gamma \in (0,1)\).
\end{definition}

\begin{remark}[Interface Function]\label{Remark-interface}Condition \eqref{eq: con2-def} in Definition \ref{def: ASF}, essentially ensures the existence of a hybrid interface function \( u = u_{\hat{u}}(x, \hat{x}, \hat{u}) \) that meets the inequality \eqref{eq: con2-def}. This interface function enables the translation of a synthesized controller \( \hat{u} \) for the symbolic model \( \hat{\Sigma} \) into a corresponding controller \( u \) for the original system \( \Sigma \).
\end{remark}

\begin{remark}[Interpretation of Definition \ref{def: ASF}]\label{Remark-insight}The ASF defines a similarity relationship between state trajectories of the dt-IANSP $\Sigma$ and its symbolic model. Specifically, if the initial states of \(\Sigma\) and \(\hat{\Sigma}\) are sufficiently close, as quantified by \(\pmb{\mathcal{S}}(x, \hat{x})\) in \eqref{eq: con1-def}, their trajectories will remain in proximity as the systems evolve over time, as described by the right-hand side terms of the inequality \eqref{eq: con2-def} \cite{tabuada2009verification}.
\end{remark}

The following theorem leverages an ASF to measure the distance between the state trajectories of a dt-IANSP \(\Sigma\) and its symbolic model \(\hat{\Sigma}\).

\begin{theorem}[Closeness Guarantee]\label{thm: model-based}Consider a dt-IANSP  \(\Sigma = (A, \mathcal{M}, B, X, U)\) and its symbolic model \(\hat{\Sigma} = (\hat{f}, \hat{X}, \hat{U})\). Assume that $\pmb{\mathcal{S}}$ is an ASF from $\hat{\Sigma}$ to $\Sigma$ as defined in Definition \ref{def: ASF}, and there exists $\nu \in \Rp$ such that $\Vert \hat u \Vert \leq \nu, \forall \hat u \in \hat U$. Then, a relation $\mathscr{R} \subseteq X \times \hat{X}$, which is defined as~\cite{tabuada2009verification}
	\begin{subequations}\label{eq: relation & error}
		\begin{align}\label{eq: relation}
			\mathscr{R} \coloneq \big\{\!(x, \hat{x}) \in X \times \hat{X} \mid \pmb{\mathcal{S}}(x, \hat{x}) \leq \max \big(\bar{\rho}\nu, \bar\psi\big)\!\big\}\!,
		\end{align}
		is an $\epsilon$-approximate alternating simulation relation from $\hat{\Sigma}$ to $\Sigma$ with
		\begin{align}\label{eq: error}
			\epsilon=\left(\frac{\max \big(\bar{\rho}\nu, \bar\psi\big)}{\alpha}\right)^{\!\!\tfrac{1}{2}}\!\!,\quad \text{where} ~~ \bar{\rho} = \frac{(1 + \eta_2) \eta_3 }{(1 - \gamma)\eta_1}\rho, ~~ \bar{\psi} = \frac{(1 + \eta_2) \eta_3 }{(1-\gamma)(\eta_3 - 1) \eta_1 \eta_2}\psi,
		\end{align}
	\end{subequations} 
	for any  $\eta_1, \eta_2 \in (0, 1)$ and $\eta_3 \in (1, 2)$.
\end{theorem}

\begin{proof}
	The proof consists of two main parts: (i) $\forall (x,\hat x) \in \mathscr{R}$ one has 
	$\Vert x-\hat x\Vert\leq  \epsilon$, and (ii) $\forall (x,\hat x) \in \mathscr{R} \text{ and } \forall \hat u\in \hat U, \exists u \in U,\text{ such that } \forall x'\in A\mathcal{M}(x) + Bu, \exists \hat x'\in \Pi(A\mathcal{M}(\hat{x}) + B\hat{u}),$ fulfilling $(x',\hat x') \in \mathscr{R}$. The first part follows directly from condition \eqref{eq: con1-def} and the definition of the relation \(\mathscr{R}\) in \eqref{eq: relation}:
	\begin{align*}
		\alpha\Vert x-\hat x\Vert ^2\leq \pmb{\mathcal{S}}(x,\hat x) \leq \max \big(\bar{\rho}\nu, \bar\psi\big)  \quad \to \quad \Vert x-\hat x\Vert \leq \left(\frac{\max \big(\bar{\rho}\nu, \bar\psi\big)}{\alpha}\right)^{\!\!\tfrac{1}{2}} = \epsilon.
	\end{align*}
	We now move on to demonstrate the second part. Since
	\begin{align*}
		\pmb{\mathcal{S}}\big(A\mathcal{M}(x) + Bu, \Pi(A\mathcal{M}(\hat{x}) + B\hat{u})\big) \leq \gamma \pmb{\mathcal{S}}(x, \hat{x}) + \rho \|\hat{u}\| + \psi \leq \max\{\bar\gamma \pmb{\mathcal{S}}(x,\hat x), \bar{\rho}\nu, \bar\psi \},
	\end{align*}
	with $\bar\gamma, \bar{\rho}$ and $\bar\psi$ as
	\begin{align*}
		\bar\gamma = 1 - (1 - \eta_1)(1 - \gamma),\quad \bar{\rho} = \frac{(1 + \eta_2) \eta_3}{(1 - \gamma)\eta_1}\rho, \quad \bar{\psi} = \frac{(1 + \eta_2) \eta_3 }{(1-\gamma)(\eta_3 - 1) \eta_1 \eta_2}\psi,
	\end{align*} 
	for any  $\eta_1, \eta_2 \!\in\! (0, 1)$ and $\eta_3 \!\in\! (1, 2)$, one has $\pmb{\mathcal{S}}(x', \hat{x}') \leq \max \big(\bar{\rho}\nu, \bar\psi\big)$ given that $\bar \gamma\!\in\!(0,1)$ and $\pmb{\mathcal{S}}(x, \hat{x}) \leq \max \big(\bar{\rho}\nu, \bar\psi\big)$ according to \eqref{eq: relation}, signifying that $(x',\hat x') \!\in\! \mathscr{R}$, which completes the proof.
\end{proof}

\begin{remark}[\textbf{Special Case $\rho \equiv 0$}]\label{remark-rho}It is worth noting that in the case of $\rho\equiv0$ (cf. Theorem \ref{thm: main}), the result of Theorem \ref{thm: model-based} is boiled down to the following:
	\begin{align}
		\epsilon = \left(\frac{\bar\psi}{\alpha}\right)^{\!\!\frac{1}{2}}\!\!,\quad \text{where} ~~ \bar\psi = \frac{\psi}{(1 - \gamma)\eta_1},
	\end{align} 
	for any  $\eta_1 \!\in\! (0, 1).$
\end{remark}

Of note is that Theorem \ref{thm: model-based} ensures the closeness between the state trajectories of a dt-IANSP and its symbolic model, enabling the enforcement of complex properties such as safety, reachability, and reach-while-avoid \cite{tabuada2009verification}. This approach allows properties to be addressed in simplified \emph{discrete-space} domains and refined for the original \emph{continuous-space} systems while maintaining a quantifiable error bound as defined in \eqref{eq: error}, leveraging the capabilities of Theorem \ref{thm: model-based}.

\subsection{Data-Driven Symbolic Models}\label{data-symbolic}

While the mathematical model of the dt-IANSP is unavailable, one can still construct the \emph{symbolic model} $\hat{f}(\hat{x}, \hat{u}) = \Pi(A\mathcal{M}(\hat{x}) + B\hat{u})$ using data by initializing the unknown dt-IANSP with a discrete state $\hat{x}$ under a discrete input $\hat{u}$ to compute $A\mathcal{M}(\hat{x}) + B\hat{u}$. With a state discretization parameter $\delta$ and by applying the quantization map $\Pi$, the symbolic model $\hat{f}(\hat{x}, \hat{u}) = \Pi(A\mathcal{M}(\hat{x}) + B\hat{u})$ is determined as the \emph{closest representative point} to the value of $A\mathcal{M}(\hat{x}) + B\hat{u}$ that satisfies condition~\eqref{eq: quant}. This procedure is repeated for all combinations of discrete state variables $\hat{x}\in \hat X$ and input variables $\hat{u}\in \hat U$, forming the data-driven symbolic model. This approach is fully aligned with the process of constructing symbolic models in model-based scenarios. It is worth noting that this way of complete symbolic model construction introduces an approximation error, represented by $\psi$ in condition~\eqref{eq: con2-def} and correspondingly in relation~\eqref{eq: relation & error}, even in the model-based case. We formally derive a closed-form expression for this error using data, based on the discretization parameter $\delta$, in Theorem~\ref{thm: main}. In particular, if $\delta$ approaches zero, the approximation error $\psi$ also goes to zero, while the cardinality of the discrete set $\hat{X}$ tends to infinity, illustrating that the discrete set increasingly resembles the continuous space.

Notwithstanding the advantages of constructing ASFs to capture the closeness between the trajectories of two systems, one needs to know the precise model of the system at hand since its dynamic appears in \eqref{eq: con2-def}. However, in many practical scenarios, the exact mathematical model of the system is not available, demonstrating the demand to develop a data-driven approach. In this regard, a few \emph{scenario-based} data-driven frameworks have been proposed to construct ASFs (\emph{e.g.,} \cite{lavaei2022data,devonport2021symbolic,lavaei2023symbolic, samari2024data}).
While promising, all these studies often require gathering a vast number of i.i.d. samples—sometimes reaching millions—which limits their practical applicability in many scenarios. Note that without formally establishing ASFs, the constructed symbolic model would be ineffective, as its results cannot be transferred back to the original system. Motivated by these key challenges, we now introduce the problem we opt to address in this work.

\begin{resp}
	\begin{problem}\label{Prob}Consider a dt-IANSP $\Sigma$ with unknown matrices $A$ and $B$, and a given upper bound on the maximum degree of $\mathcal{M}(x)$. Develop a direct data-driven method to construct an ASF, ensuring its correctness guarantees by leveraging only two input-state trajectories. Following this, design a hybrid controller to enforce the desired behaviors over the unknown dt-IANSP, using the controller synthesized for its data-driven symbolic model.
	\end{problem}
\end{resp}

In order to address Problem \ref{Prob}, we introduce our data-driven framework in the subsequent section.

\section{Data-Driven Framework}\label{Data Framework}
This section introduces our proposed data-driven approach for constructing an ASF from the data-driven symbolic model $\hat{\Sigma}$ to the unknown dt-IANSP $\Sigma$. To achieve this, we first define the ASF in a \emph{quadratic} form as \(\pmb{\mathcal{S}}(x, \hat{x}) = (x - \hat{x})^\top \mathcal{P} (x - \hat{x})\), where \(\mathcal{P} \succ 0\). Next, assuming that all system states and inputs are measurable, we gather input-state data (\emph{i.e.,} observations) over a time horizon \(\left[0, \mathcal{T}\right]\), where \(\mathcal{T} \in \Np\) denotes the number of collected samples:
\begin{subequations}\label{eq: data1}
	\begin{align}
		\pmb{\mathcal{O}} & \coloneq \begin{bmatrix}
			x(0) & x(1) & \dots & x(\mathcal{T} - 1)
		\end{bmatrix} \in \R^{n \times \mathcal{T}}\! ,\\
		\pmb{\mathcal{I}} & \coloneq \begin{bmatrix}
			u(0) & u(1) & \dots & u(\mathcal{T} - 1)
		\end{bmatrix} \in \R^{m \times \mathcal{T}}\! ,\\
		\pmb{\mathcal{O}^+} & \coloneq \begin{bmatrix}
			x(1) & x(2) & \dots & x(\mathcal{T})
		\end{bmatrix} \in \R^{n \times \mathcal{T}}\!.
	\end{align}
\end{subequations}
We refer to \eqref{eq: data1} as a \emph{single input-state  trajectory} from the system. To construct an ASF within our framework, we collect another trajectory, this time from $A \mathcal{M}(\hat{x}) + B \hat{u}$ (referred to as the \emph{second trajectory}), similar to \eqref{eq: data1}, with the same number of samples $\mathcal{T}$. This second trajectory is denoted as $\pmb{\hat{\mathcal{O}}}, \pmb{\hat{\mathcal{I}}},$ and $\pmb{\hat{\mathcal{O}}^+}$, in which the initial condition $\hat{x}(0)$ can be selected from the discrete state set $\hat{X}$ and the control inputs from the discrete input set $\hat{U}$. Gathering the second trajectory is necessary within our framework as we ultimately desire to obtain the data-based representation of $A\mathcal{M}(x) + Bu - (A \mathcal{M}(\hat{x}) + B \hat{u})$ (cf. Lemma \ref{lemma}) to construct the ASF (cf. Theorem \ref{thm: main}).

It is worth noting that $A \mathcal{M}(\hat{x}) + B \hat{u}$, from which the second trajectory is collected, differs from the symbolic model dynamics in \eqref{symbolic dyn} since it does not involve the quantization map $\Pi$. Specifically, the incremental difference $A\mathcal{M}(x) + Bu - (A \mathcal{M}(\hat{x}) + B \hat{u})$ is required, as this work aims to construct a \emph{complete} abstraction in accordance with condition \eqref{eq: quant}. This relies on the original system satisfying the incremental input-to-state stability (delta-ISS) property (cf. \cite[Definition 4.1]{swikir2019compositional}), which is implicitly reflected in condition \eqref{eq: con2-def} with $\gamma\in(0,1)$. This enables us to construct an ASF from $\hat\Sigma$ to $\Sigma$ (\(\hat{\Sigma} \preceq_{{\pmb{\mathcal{S}}}} \! \Sigma\)) as well as from $\Sigma$ to $\hat \Sigma$ (\({\Sigma} \preceq_{{\pmb{\mathcal{S}}}} \! \hat\Sigma\)), \emph{i.e.,} thereby establishing an alternating \emph{bisimulation} function between $\Sigma$ and $\hat \Sigma$ (\({\Sigma} \cong_{\pmb{\mathcal{S}}} \!\hat \Sigma\))—see Theorem~\ref{thm: main} and Remark~\ref{Rem77}. It is worth noting that the delta-ISS nature of complete abstraction in our setting implies that the second trajectory can be even collected from any arbitrary states and inputs within the continuous domain.

Having been inspired by \cite{guo2021data}, in the following lemma, we raise the required conditions along with the proposed hybrid interface function to derive the closed-form data-based representation of $A\mathcal{M}(x) + Bu - (A \mathcal{M}(\hat{x}) + B \hat{u})$, which is subsequently utilized in the main results of the paper.

\begin{lemma}[Data-based Representation]\label{lemma}Assume $\pmb{\mathcal{G}_1}(x) \in \R^{\mathcal{T} \times n}$ and $\pmb{\mathcal{G}_2}(\hat x) \in \R^{\mathcal{T} \times n}$ are arbitrary matrix-valued functions such that
	\begin{subequations} \label{eq: Eq1}
		\begin{align}
			\Upsilon(x) &= \pmb{\mathds{M}} \pmb{\mathcal{G}_1}(x), \label{eq: cond_G1}\\
			\Upsilon(\hat x) &= \pmb{\hat{\mathds{M}}} \pmb{\mathcal{G}_2}(\hat x), \label{eq: cond_G2}
		\end{align}
	\end{subequations}
	where $\Upsilon(\cdot) \in \R^{M \times n}$ is a matrix-valued function satisfying the following conditions:
	\begin{subequations}\label{eq: Eq2}
		\begin{align}
			\mathcal{M}(x) &= \Upsilon(x) x, \label{eq: T1}\\
			\mathcal{M}(\hat x) &= \Upsilon(\hat x) \hat x. \label{eq: T2}
		\end{align}
	\end{subequations}
	Moreover, $\pmb{\mathds{M}}$ is a \emph{full row-rank} matrix constructed from the vector $\mathcal{M}(x)$ and sampled data $\pmb{\mathcal{O}}$ as
	\begin{align}
		\pmb{\mathds{M}} = \begin{bmatrix}
			\mathcal{M}(x(0)) & \mathcal{M}(x(1)) & \dots & \mathcal{M}(x(\mathcal{T} - 1))
		\end{bmatrix} \in \R^{M \times \mathcal{T}}\! , \label{eq: M of M}
	\end{align}
	and $\pmb{\hat{\mathds{M}}}$ is constructed analogously. If the hybrid interface function is designed as
	\begin{align}
		u = \pmb{\mathcal{I}}  \pmb{\mathcal{G}_1}(x) x - \pmb{\hat{\mathcal{I}}}  \pmb{\mathcal{G}_2}(\hat x)\hat x + \hat u, \label{eq: interface}
	\end{align}
	then the difference $A\mathcal{M}(x) + Bu - (A \mathcal{M}(\hat{x}) + B \hat{u})$ can be expressed in the data-based form
	\begin{align}
		A\mathcal{M}(x) + Bu - (A \mathcal{M}(\hat{x}) + B \hat{u}) = \pmb{\mathcal{O}^+} \pmb{\mathcal{G}_1}(x) x - \pmb{\hat{\mathcal{O}}^+} \pmb{\mathcal{G}_2}(\hat x) \hat{x}. \label{eq: closed}
	\end{align}
\end{lemma}

\begin{proof}
	We start with the fact that according to the data collected in \eqref{eq: data1} and the matrix $\pmb{\mathds{M}}$, one has
	\begin{align}
		\pmb{\mathcal{O}^+} = A \pmb{\mathds{M}} + B \pmb{\mathcal{I}}. \label{tmp1}
	\end{align}
	Likewise, we can show that $\pmb{\hat{\mathcal{O}}^+} = A \pmb{\hat{\mathds{M}}} + B \pmb{\hat{\mathcal{I}}}$. Now, one can derive the following chain of equalities:
	\begin{align*}
		A\mathcal{M}(x) \!+\! Bu \!-\! (A \mathcal{M}(\hat{x}) \!+\! B \hat{u}) &\overset{\eqref{eq: Eq2}}{=} A \Upsilon(x) x + Bu - (A \Upsilon(\hat x) \hat x + B \hat u)\\
		& \! \overset{\eqref{eq: interface}}{=} A \Upsilon(x) x + B \pmb{\mathcal{I}}  \pmb{\mathcal{G}_1}(x) x  - B \pmb{\hat{\mathcal{I}}}  \pmb{\mathcal{G}_2}(\hat x)\hat x + B \hat u - (A \Upsilon(\hat x) \hat x \!+\! B \hat u)\\
		&\hspace{0.1cm} = A \Upsilon(x) x + B \pmb{\mathcal{I}}  \pmb{\mathcal{G}_1}(x) x - (A \Upsilon(\hat x) \hat x + B \pmb{\hat{\mathcal{I}}}  \pmb{\mathcal{G}_2}(\hat x)\hat x)\\
		& \overset{\eqref{eq: Eq1}}{=} (A \pmb{\mathds{M}} + B \pmb{\mathcal{I}}) \pmb{\mathcal{G}_1}(x) x - (A \pmb{\hat{\mathds{M}}} + B \pmb{\hat{\mathcal{I}}}) \pmb{\mathcal{G}_2}(\hat x)\hat x\\
		& \! \overset{\eqref{tmp1}}{=} \pmb{\mathcal{O}^+} \pmb{\mathcal{G}_1}(x) x - \pmb{\hat{\mathcal{O}}^+} \pmb{\mathcal{G}_2}(\hat x)\hat x,
	\end{align*}
	which concludes the proof.
\end{proof}

\begin{remark}[Rank Conditions]\label{Remark-richness}If the matrices $\pmb{\mathds{M}}$ and $\pmb{\hat{\mathds{M}}}$ possess full-row rank, this ensures the existence of matrices $\pmb{\mathcal{G}_1}(x)$ and $\pmb{\mathcal{G}_2}(\hat x)$ that satisfy \eqref{eq: Eq1}. To fulfill this condition, it is necessary to collect at least $M + 1$ data points, implying that $\mathcal{T}$ must satisfy $\mathcal{T} \geq M + 1$. Note that this requirement is readily verifiable since matrices $\pmb{\mathds{M}}$ and $\pmb{\hat{\mathds{M}}}$ are constructed from the gathered data.
\end{remark}

\begin{remark}[Transformation Matrices]Since $\mathcal{M}(\mathbf{0}_n) = \mathbf{0}_M$, it is evident that, without loss of generality, transformation matrices $\Upsilon(x)$ and $\Upsilon(\hat x)$ can always be found to satisfy conditions~\eqref{eq: Eq2}. This approach facilitates our framework by representing everything in terms of $x$ and $\hat{x}$ instead of $\mathcal{M}(x)$ and $\mathcal{M}(\hat{x})$, aligning with the structure of our ASF  $\pmb{\mathcal{S}}(x, \hat{x}) = (x - \hat{x})^\top \mathcal{P} (x - \hat{x})$.
\end{remark}

With the data-based representation of $A\mathcal{M}(x) + Bu - (A \mathcal{M}(\hat{x}) + B \hat{u})$ established in Lemma \ref{lemma}, we now present the following theorem as the main contribution of this work. This theorem enables the design of an ASF from the data-driven symbolic model $\hat{\Sigma}$ to the original dt-IANSP $\Sigma$, relying solely on two trajectories from the unknown system.

\begin{theorem}[Data-Driven ASFs]\label{thm: main}Consider an unknown dt-IANSP $\Sigma$, defined in Definition \ref{def: sys}, its data-driven symbolic model $\hat{\Sigma}$, defined in Subsection \ref{data-symbolic}, and the data-based representation of $A\mathcal{M}(x) + Bu - A \mathcal{M}(\hat{x}) - B \hat{u}$, as in \eqref{eq: closed}, proposed in Lemma \ref{lemma}. Now, if there exist matrix-valued functions $\pmb{\mathcal{Y}_1}(x) \in \R^{\mathcal{T} \times n}$ and $\pmb{\mathcal{Y}_2}(\hat x) \in \R^{\mathcal{T} \times n}$, along with constant matrices $\Xi \in \R^{n \times n}$ and $\Theta \in \R^{n \times n},$ with $\Xi \succ 0,$ satisfying
	\begin{subequations}\label{eq: thm}
		\begin{align}
			&\pmb{\mathds{M}}  \pmb{\mathcal{Y}_1}(x) = \Upsilon(x) \Xi, \label{eq: thm1}\\
			& \pmb{\hat{\mathds{M}}} \pmb{\mathcal{Y}_2}(\hat x) =  \Upsilon(\hat x) \Xi, \label{eq: thm2}\\
			& \pmb{\mathcal{O}^+} \pmb{\mathcal{Y}_1}(x) = \Theta, \label{eq: thm3}\\
			& \pmb{\hat{\mathcal{O}}^+}  \pmb{\mathcal{Y}_2}(\hat x) = \Theta, \label{eq: thm4}\\
			& \begin{bmatrix}
				(\tfrac{1}{1 + \mu})\Xi & \Theta\\
				\star & \gamma\Xi
			\end{bmatrix} \succeq 0, \label{eq: thm5}
		\end{align}
	\end{subequations}
	for some $\mu \in \Rp$ and \(\gamma \in (0,1)\), then one can conclude that $\pmb{\mathcal{S}}(x, \hat x) = (x - \hat{x})^\top \mathcal{P} (x - \hat{x}),$ with $\mathcal{P} \coloneq \Xi^{-1},$ is an ASF from the data-driven symbolic model $\hat{\Sigma}$ to the unknown dt-IANSP $\Sigma$, with $\alpha = \lambda_{\min}(\mathcal{P}), \rho \equiv 0,$ and $\psi = (1 + \frac{1}{\mu}) \Vert \sqrt{\mathcal{P}} \Vert^2 \delta^2.$ Furthermore, the interface function is designed as $u =  \pmb{\mathcal{I}}  \pmb{\mathcal{Y}_1}(x) \mathcal{P} x - \pmb{\hat{\mathcal{I}}}  \pmb{\mathcal{Y}_2}(\hat x) \mathcal{P} \hat x + \hat u.$
\end{theorem}

\begin{proof}
	Initially, let us define 
	\begin{align}\label{New9}
		\pmb{\mathcal{G}_1}(x) \coloneq \pmb{\mathcal{Y}_1}(x) \mathcal{P}, \quad\text{and} \quad \pmb{\mathcal{G}_2}(\hat x) \coloneq \pmb{\mathcal{Y}_2}(\hat x) \mathcal{P}.
	\end{align}
	Since $\mathcal{P} = \Xi^{-1},$ we consequently get $\pmb{\mathcal{G}_1}(x) \coloneq \pmb{\mathcal{Y}_1}(x) \Xi^{-1}$ and $\pmb{\mathcal{G}_2}(\hat x) \coloneq \pmb{\mathcal{Y}_2}(\hat x) \Xi^{-1}.$ Thus, the satisfaction of \eqref{eq: cond_G1} and \eqref{eq: cond_G2} directly follows from the fulfillment of \eqref{eq: thm1} and \eqref{eq: thm2}, respectively.
	
	Now, we proceed with demonstrating the fulfillment of condition \eqref{eq: con2-def} under the satisfaction of conditions \eqref{eq: thm1}-\eqref{eq: thm5}. To do so, by considering $\pmb{\mathcal{S}}(x, \hat x) = (x - \hat{x})^\top \mathcal{P} (x - \hat{x}),$ one has
	\begin{align*}
		&\pmb{\mathcal{S}}\big(A\mathcal{M}(x) + Bu, \Pi(A\mathcal{M}(\hat{x}) + B\hat{u})\big)\\
		& = \big(\underbrace{A\mathcal{M}(x) + Bu - \Pi(A\mathcal{M}(\hat{x}) + B\hat{u})}_\clubsuit\big)^{\!\! \top} \mathcal{P} \big(A\mathcal{M}(x) + Bu - \Pi(A\mathcal{M}(\hat{x}) + B\hat{u})\big).
	\end{align*}
	We first aim to compute a closed-form data-based representation for ``$\clubsuit$'' as follows:
	\begin{align*}
		&A\mathcal{M}(x) + Bu - \Pi(A\mathcal{M}(\hat{x}) + B\hat{u})\\
		& \! \overset{\eqref{eq: interface}}{=} A\mathcal{M}(x) + B \pmb{\mathcal{I}}  \pmb{\mathcal{G}_1}(x) x - B \pmb{\hat{\mathcal{I}}}  \pmb{\mathcal{G}_2}(\hat x)\hat x + B \hat u - \Pi(A\mathcal{M}(\hat{x}) + B\hat{u}).
	\end{align*}
	Now, by incorporating the term $A\mathcal{M}(\hat{x})$ through \emph{addition and subtraction}, we have
	\begin{align*}
		&A\mathcal{M}(x) + Bu - \Pi(A\mathcal{M}(\hat{x}) + B\hat{u})\\
		& =  A\mathcal{M}({x}) + B \pmb{\mathcal{I}}  \pmb{\mathcal{G}_1}(x) x - A\mathcal{M}(\hat{x}) - B \pmb{\hat{\mathcal{I}}}  \pmb{\mathcal{G}_2}(\hat x)\hat x + A\mathcal{M}(\hat{x}) + B \hat u - \Pi(A\mathcal{M}(\hat{x}) + B\hat{u}). 
	\end{align*}
	Subsequently, according to \eqref{eq: closed} in Lemma \ref{lemma}, we have
	\begin{align*}
		&A\mathcal{M}(x) + Bu - \Pi(A\mathcal{M}(\hat{x}) + B\hat{u})\\
		& \overset{\eqref{eq: closed}}{=}  \pmb{\mathcal{O}^+} \pmb{\mathcal{G}_1}(x) x - \pmb{\hat{\mathcal{O}}^+} \pmb{\mathcal{G}_2}(\hat x)\hat x + A\mathcal{M}(\hat{x})+ B \hat u - \Pi(A\mathcal{M}(\hat{x})+ B\hat{u})\\
		& \overset{\eqref{New9}}{=}   \pmb{\mathcal{O}^+} \pmb{\mathcal{Y}_1}(x) \mathcal{P} x - \pmb{\hat{\mathcal{O}}^+} \pmb{\mathcal{Y}_2}(\hat x) \mathcal{P} \hat x + A\mathcal{M}(\hat{x}) + B \hat u - \Pi(A\mathcal{M}(\hat{x}) + B\hat{u}).
	\end{align*}
	Now, according to conditions \eqref{eq: thm3} and \eqref{eq: thm4}, one has
	\begin{align}\notag
		&A\mathcal{M}(x) + Bu - \Pi(A\mathcal{M}(\hat{x}) + B\hat{u})\\\label{new6}
		& = \Theta \mathcal{P} (x- \hat x)  + A\mathcal{M}(\hat{x})+ B \hat u - \Pi(A\mathcal{M}(\hat{x}) + B\hat{u}).
	\end{align}
	Then by having the closed-form data-based representation of ``$\clubsuit$'' in \eqref{new6}, we have 
	\begin{align}\notag
		&\pmb{\mathcal{S}}\big(A\mathcal{M}(x) + Bu, \Pi(A\mathcal{M}(\hat{x}) + B\hat{u})\big)\\\notag
		& = \big(\Theta \mathcal{P} (x- \hat x)  + A\mathcal{M}(\hat{x}) + B \hat u - \Pi(A\mathcal{M}(\hat{x}) + B\hat{u})\big)^{\!\! \top}\mathcal{P} \big(\Theta \mathcal{P} (x- \hat x)  + A\mathcal{M}(\hat{x}) + B \hat u \\\notag
		& \hspace{0.7cm} - \Pi(A\mathcal{M}(\hat{x}) + B\hat{u})\big)\\\notag
		& = (x - \hat x)^{\! \top} \mathcal{P} \Theta^{\! \top}  \mathcal{P} \Theta \mathcal{P} (x - \hat x) + 2 \overbrace{(x - \hat x)^{\! \top} \mathcal{P} \Theta^{\! \top} \sqrt{\mathcal{P}}}^{a} \overbrace{\sqrt{\mathcal{P}} (A\mathcal{M}(\hat{x}) + B \hat u - \Pi(A\mathcal{M}(\hat{x}) + B\hat{u}))}^{b}\\\label{New62}
		& \hspace{0.3cm} + (A\mathcal{M}(\hat{x}) + B \hat u - \Pi(A\mathcal{M}(\hat{x}) + B\hat{u}))^{\! \top} \mathcal{P} (A\mathcal{M}(\hat{x}) + B \hat u - \Pi(A\mathcal{M}(\hat{x})+ B\hat{u})).
	\end{align}
	According to the Cauchy-Schwarz inequality \cite{bhatia1995cauchy}, \emph{i.e.,}  $a b \leq \Vert a \Vert \Vert b \Vert,$ for any $a^\top\!, b \in \R^{n}$, followed by
	employing Young's inequality \cite{young1912classes}, \emph{i.e.,} $\Vert a \Vert \Vert b \Vert \leq \frac{\mu}{2} \Vert a \Vert^2 + \frac{1}{2\mu} \Vert b \Vert^2$, for any $\mu \in \Rp$, one has
	\begin{align}\notag
		& 2 (x - \hat x)^{\! \top} \mathcal{P} \Theta^{\! \top} \sqrt{\mathcal{P}} \sqrt{\mathcal{P}} (A\mathcal{M}(\hat{x})+ B \hat u - \Pi(A\mathcal{M}(\hat{x}) + B\hat{u}))\\\notag
		& \leq \mu (x - \hat x)^{\! \top} \mathcal{P} \Theta^{\! \top}  \mathcal{P} \Theta \mathcal{P} (x - \hat x) + \tfrac{1}{\mu} (A\mathcal{M}(\hat{x}) + B \hat u - \Pi(A\mathcal{M}(\hat{x}) + B\hat{u}))^{\! \top} \mathcal{P} (A\mathcal{M}(\hat{x}) + B \hat u\\\label{New6}
		& \hspace{0.3cm} - \Pi(A\mathcal{M}(\hat{x}) + B\hat{u})).
	\end{align}
	Likewise, using Cauchy-Schwarz inequality, one can conclude that
	\begin{align}\notag
		&(A\mathcal{M}(\hat{x}) + B \hat u - \Pi(A\mathcal{M}(\hat{x}) + B\hat{u}))^{\! \top} \mathcal{P} (A\mathcal{M}(\hat{x}) + B \hat u - \Pi(A\mathcal{M}(\hat{x}) + B\hat{u}))\\\label{New61}
		&\leq \Vert \sqrt{\mathcal{P}} \Vert ^2 \Vert (A\mathcal{M}(\hat{x}) + B \hat u - \Pi(A\mathcal{M}(\hat{x}) + B\hat{u})) \Vert^2 \overset{\eqref{eq: quant}}{\leq} \Vert \sqrt{\mathcal{P}} \Vert ^2 \delta^2.
	\end{align}
	Hence, by applying the bounds \eqref{New6} and \eqref{New61} to \eqref{New62}, one has
	\begin{align*}
		\pmb{\mathcal{S}}\big(A\mathcal{M}(x) + Bu, \Pi(A\mathcal{M}(\hat{x}) + B\hat{u})\big) \leq (1 + \mu) (x - \hat x)^{\! \top} \mathcal{P} \Theta^{\! \top}  \mathcal{P} \Theta \mathcal{P} (x - \hat x) + (1 + \tfrac{1}{\mu}) \Vert \sqrt{\mathcal{P}} \Vert ^2 \delta^2.
	\end{align*}
	According to the Schur complement \cite{zhang2006schur}, and considering  condition \eqref{eq: thm5} with $\mathcal{P} = \Xi^{-1}$, it is well-defined that
	\begin{align*}
		\begin{bmatrix}
			(\tfrac{1}{1 + \mu})\mathcal{P}^{-1} & \Theta\\
			\star & \gamma \mathcal{P}^{-1}
		\end{bmatrix} \succeq 0& \Leftrightarrow \gamma \mathcal{P}^{-1} -(1 + \mu) \Theta^{\! \top} \mathcal{P} \Theta \succeq 0 \Leftrightarrow \gamma \underbrace{\mathcal{P}\mathcal{P}^{-1}}_{\I_n}\mathcal{P} -  (1 + \mu)\mathcal{P} \Theta^{\! \top} \mathcal{P} \Theta  \mathcal{P} \succeq 0.
	\end{align*}
	Thus, it is clear that
	\begin{align*}
		(1 + \mu) (x - \hat x)^{\! \top} \mathcal{P} \Theta^{\! \top}  \mathcal{P} \Theta \mathcal{P} (x - \hat x) \leq \gamma (x - \hat x)^{\! \top} \mathcal{P} (x - \hat x).
	\end{align*}
	Hence, one can deduce that 
	\begin{align*}
		\pmb{\mathcal{S}}\big(A\mathcal{M}(x) + Bu, \Pi(A\mathcal{M}(\hat{x}) + B\hat{u})\big) \leq \gamma \underbrace{(x - \hat x)^{\! \top} \mathcal{P} (x - \hat x)}_{\pmb{\mathcal{S}}(x, \hat{x})} + \underbrace{(1 + \tfrac{1}{\mu}) \Vert \sqrt{\mathcal{P}} \Vert ^2 \delta^2}_\psi,
	\end{align*}
	implying that condition \eqref{eq: con2-def} is satisfied under the fulfillment of conditions \eqref{eq: thm1}-\eqref{eq: thm5}, with $ \rho \equiv 0$ and $\psi = (1 + \frac{1}{\mu}) \Vert \sqrt{\mathcal{P}} \Vert^2 \delta^2.$ To complete the proof, we note that
	\begin{align*}
		\lambda_{\min}(\mathcal{P})  \|x - \hat{x}\|^2 \leq \pmb{\mathcal{S}}(x, \hat{x}) \leq \lambda_{\max}(\mathcal{P})  \|x - \hat{x}\|^2,
	\end{align*}
	and thus, condition \eqref{eq: con1-def} is satisfied with $\alpha =  \lambda_{\min}(\mathcal{P}),$ thereby concluding the proof.
\end{proof}

\begin{remark}[Enforcement of Conditions  \eqref{eq: thm1}-\eqref{eq: thm5}]To efficiently impose the conditions  \eqref{eq: thm1}-\eqref{eq: thm5}, one can leverage existing software tools such as \textsf{SOSTOOLS} \cite{prajna2004sostools} for fulfilling conditions~\eqref{eq: thm1}-\eqref{eq: thm4} in conjunction with semi-definite programming solvers like \textsf{SeDuMi} \cite{sturm1999using} for solving condition~\eqref{eq: thm5}.
\end{remark}

\begin{remark}[ASF from $\Sigma$ to $\hat \Sigma$]\label{Rem77}While the results of Theorem \ref{thm: main} construct an ASF from $\hat\Sigma$ to $\Sigma$ (\(\hat{\Sigma} \preceq_{{\pmb{\mathcal{S}}}} \! \Sigma\)), one can readily show that the proposed $\pmb{\mathcal{S}}(x, \hat{x})$ is also an ASF from $\Sigma$ to $\hat \Sigma$ (\({\Sigma} \preceq_{{\pmb{\mathcal{S}}}} \! \hat\Sigma\)), thereby establishing an alternating \emph{bisimulation} function between $\Sigma$ and $\hat \Sigma$ (\({\Sigma} \cong_{{\pmb{\mathcal{S}}}} \!\hat \Sigma\)). In particular, by the definition of $\hat U$, for any $u =  \pmb{\mathcal{I}}  \pmb{\mathcal{Y}_1}(x) \mathcal{P} x - \pmb{\hat{\mathcal{I}}}  \pmb{\mathcal{Y}_2}(\hat x) \mathcal{P} \hat x + \tilde u \in U$ there always exists $\hat u \in \hat U$ such that $\Vert\tilde u - \hat u \Vert\leq  \beta_1$, with $\beta_1$ being the input discretization parameter. If one follows the proof of Theorem \ref{thm: main} with the new $u =  \pmb{\mathcal{I}}  \pmb{\mathcal{Y}_1}(x) \mathcal{P} x - \pmb{\hat{\mathcal{I}}}  \pmb{\mathcal{Y}_2}(\hat x) \mathcal{P} \hat x + \tilde u$, and through addition and subtraction of $A\mathcal{M}(\hat{x}) + B\hat u$ (instead of $A\mathcal{M}(\hat{x})$),  all proof steps remain valid under the satisfaction of conditions \eqref{eq: thm1}-\eqref{eq: thm5}, with the only modification being the parameter $\psi$, now constructed as $\psi = (1 + \frac{1}{\mu}) \Vert \sqrt{\mathcal{P}} \Vert^2 (\beta_1^2\beta_2^2 + \delta^2)$, with $\beta_2$ being an upper bound on the induced $2$-norm of $B$ (cf. \cite[Theorem 4.8]{swikir2019compositional} for the model-based analysis).
\end{remark}

One can use the data-driven symbolic models  to synthesize discrete controllers $\hat u$ capable of satisfying complex logic properties via formal methods tools, such as SCOTS \cite{rungger2016scots}.
Leveraging the ASF, established by Theorem~\ref{thm: main}, and its corresponding data-driven interface map, expressed as
$u =  \pmb{\mathcal{I}}  \pmb{\mathcal{Y}_1}(x) \mathcal{P} x - \pmb{\hat{\mathcal{I}}}  \pmb{\mathcal{Y}_2}(\hat x) \mathcal{P} \hat x + \hat u$,
the hybrid controller $u$ can be effectively applied to the unknown dt-IANSP $\Sigma$. This ensures that the state trajectories of the two systems remain closely aligned, as guaranteed by the results of Theorem~\ref{thm: model-based}.

\section{Simulation Results}\label{Simul}
In this section, we evaluate the efficacy of our data-driven framework by constructing an ASF and a symbolic model using data, while enforcing two types of specifications over an unknown dt-IANSP: (i) safety and (ii) reach-while-avoid. All simulations were performed in Matlab on a MacBook Pro (Apple M2 Max) with 32GB memory.

Consider the following dt-IANSP \cite{guo2020learning}
\begin{align}
	\Sigma\!: \begin{cases}
		x_1^+ = x_1 + \tau x_2,\\
		x_2^+ = x_2 + \tau (x_1^2 + u),
	\end{cases}\label{sys_org}
\end{align}
where $\tau = 0.2$ is the sampling time. We note that the dt-IANSP \eqref{sys_org} is in the form of \eqref{eq: dt-IANSP}, with
\begin{align*}
	A = \begin{bmatrix}
		1 & \tau & 0\\
		0 & 1 & \tau
	\end{bmatrix}\!\!, \quad \mathcal{M}(x) = \begin{bmatrix}
		x_1 & x_2 & x_1^2
	\end{bmatrix}^{\! \top}\!\!\!\!, \quad B = \begin{bmatrix}
		0\\
		\tau
	\end{bmatrix}\!\!.
\end{align*}
We stress that matrices $A$ and $B$, along with the vector of monomials $\mathcal{M}(x)$, are all unknown. However, we assume that the maximum degree of the vector of monomials, $\mathcal{M}(x)$, is given as $2$. Based on this, we construct the extended dictionary $\mathcal{M}(x) = \begin{bmatrix}
	x_1 & x_2 & x_1^2 & x_1 x_2 & x_2^2
\end{bmatrix}^{\! \top}$\!\!\!.

The main goal is to collect \emph{two input-state trajectories} from the unknown system to design an ASF and a \emph{hybrid} interface map based on the result of Theorem \ref{thm: main}. Subsequently, as the first safety specification, the data-driven symbolic model is employed to design a discrete controller within the input set $U = [-2.5, 2.5]$, ensuring that the system's states remain within a predefined safe set $X = [-0.5, \: 0.5] \times [-0.5, \: 0.5]$ with correctness guarantees. To do so, we collect two sets of trajectories from the unknown dt-IANSP \eqref{sys_org}, as in \eqref{eq: data1}, over the time horizon $\mathcal{T} = 9$. Moreover, we set $\mu = 0.01, \delta = 0.001,$ and $\gamma = 0.99$. By solving conditions \eqref{eq: thm1}-\eqref{eq: thm5}, we design
\begin{align*}
	\mathcal{P} = \begin{bmatrix}
		5.5818  &  1.6117\\
		1.6117  &  6.1267
	\end{bmatrix}\!\!, \quad \Theta = \begin{bmatrix}
		0.1837 &  -0.0157\\
		-0.0487 &   0.0135
	\end{bmatrix}\!\!,
\end{align*}
and the hybrid interface map 
\begin{align}\label{int}
	u = -x_1^2 - 1.2495x_1 - 4.9773x_2 + \hat x_1^2 + 1.2495\hat x_1 + 4.9773\hat x_2 + \hat{u}.
\end{align}
Furthermore, based on the designed values, we compute $\psi = 0.0014$ and $\alpha = 4.2197$. Finally, with $\eta_1 = 0.99$, we compute $\bar{\psi} = 0.1414$ and accordingly $\epsilon = 0.1831$ based on Remark~\ref{remark-rho}, which formally quantifies the closeness between the state trajectories of the data-driven symbolic model and the unknown dt-IANSP.

\begin{figure}[t!]
	\centering
	\subfloat[Trajectories of the original system and its data-driven symbolic model]{
		\includegraphics[width=0.41\textwidth]{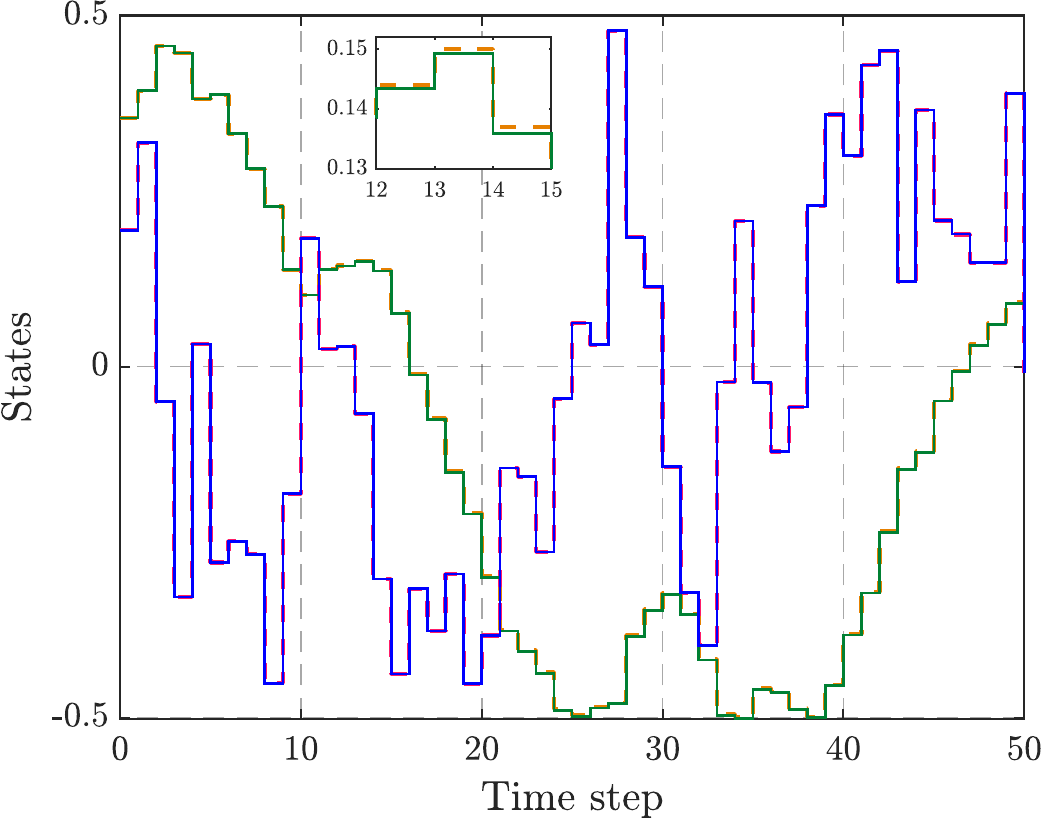}
		\label{fig:sub1}
	}
	\hfill
	\subfloat[Norm of the error between two systems' trajectories]{
		\includegraphics[width=0.41\textwidth]{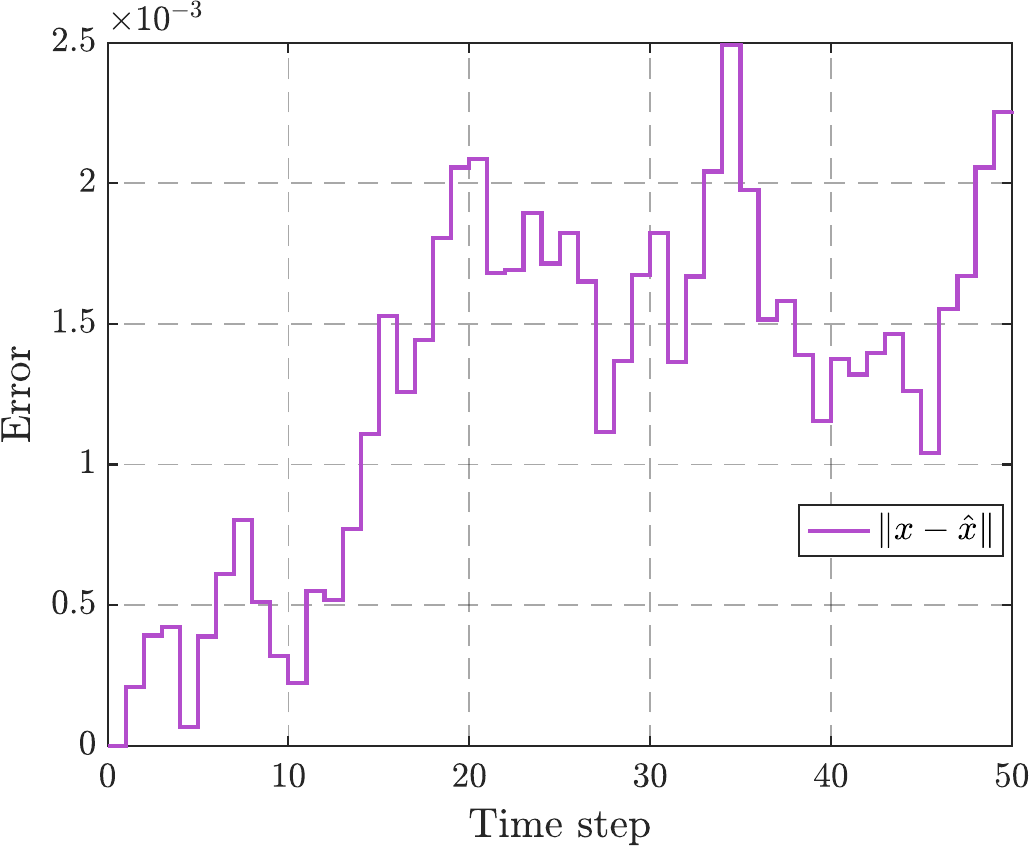}
		\label{fig:sub2}
	}
	\caption{The left-hand side subfigure illustrates the trajectories of \(x_1\) (\sampleline{x1, thick}), \(x_2\) (\sampleline{x2, thick}), \(\hat{x}_1\) (\sampleline{dashed, xh1, thick}), and \(\hat{x}_2\) (\sampleline{dashed, xh2, thick}). As shown, it is clear that the safety property is fulfilled. Moreover, the right-hand side subfigure depicts the norm of the error between the trajectories of the original system and those of its symbolic model. As illustrated, the error is always lower than the $\epsilon$ reported.}
	\label{fig:main}
\end{figure}

Using the constructed data-driven symbolic model, we proceed to design a controller for the dt-IANSP \eqref{sys_org} that ensures the state variables remain within the safe set $X = [-0.5, \: 0.5] \times [-0.5, \: 0.5]$. To achieve this, we first synthesize a discrete controller for the data-driven symbolic model using SCOTS \cite{rungger2016scots}. Subsequently, this controller is refined back over the unknown dt-IANSP using the data-driven interface map designed in~\eqref{int}. The closed-loop state trajectories of the unknown system using the hybrid interface map are depicted in Figure~\ref{fig:sub1}, while the error between the state trajectories of the unknown system and its symbolic model is shown in Figure~\ref{fig:sub2}. As shown, the synthesized hybrid controller successfully enforces the state variables of the unknown dt-IANSP \eqref{sys_org} to remain within the safe set $X = [-0.5, \: 0.5] \times [-0.5, \: 0.5]$.

To further demonstrate the efficacy of our framework, we opt to enforce a \emph{reach-while-avoid} specification over the unknown system \eqref{sys_org} without redesigning the interface function \eqref{int}. This highlights that once the ASF, interface function, and data-driven symbolic model are established, enforcing different specifications on the unknown system becomes readily achievable. The \emph{reach-while-avoid} specification requires the system's position to start within the initial set $[-1, -0.6] \times [-1, -0.5]$, reach the target set $[0.7, 1] \times [0.7, 1]$, while avoiding the obstacle $[-0.5, 0.5] \times [-1, 0.5]$. To achieve this, we utilize the constructed data-driven symbolic model and synthesize a discrete controller employing SCOTS, followed by refining it back over the unknown system leveraging the hybrid interface function \eqref{int}. Figure~\ref{fig:Rwa} provides the simulation result, which shows the practicality of our framework for enforcing such a complex specification.

\begin{figure}[t!]
	\centering
	\includegraphics[width=0.42\linewidth]{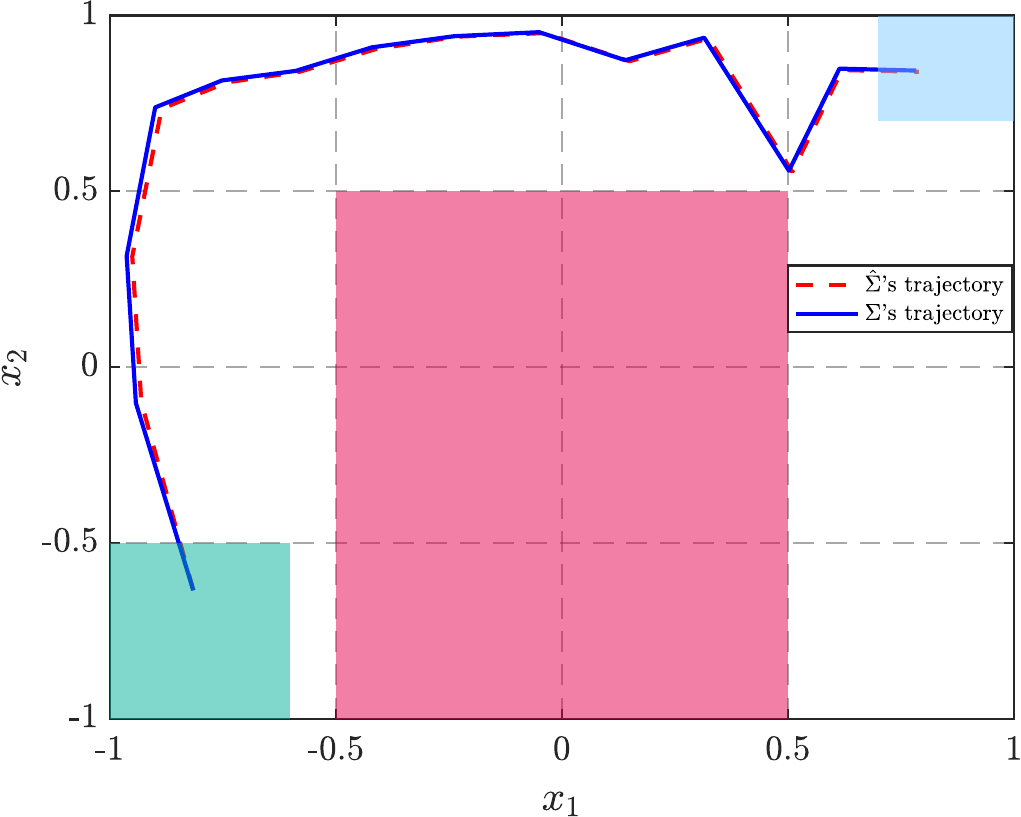}
	\caption{The trajectories must originate from the initial set \legendsquare{START!50}, reach the target set \legendsquare{TARGET!50}, while avoiding the obstacle \legendsquare{OBSTACLES!50}. As shown, the trajectories of the unknown system \eqref{sys_org} (\sampleline{X1, thick}) and its data-driven symbolic model (\sampleline{dashed, X2, thick}) closely align, demonstrating compliance with the system's reach-while-avoid property.}
	\label{fig:Rwa}
\end{figure}

\section{Conclusion}\label{Conc}
In this work, we introduced a data-driven abstraction-based methodology that synthesizes a hybrid controller for unknown nonlinear polynomial systems using only \emph{two input-state trajectories}. Our framework ensures that the desired system behavior is achieved by designing a controller derived from a symbolic model. To achieve this, we utilized alternating simulation functions (ASFs) to establish a formal relationship between the state trajectories of the unknown system and its symbolic counterpart. Our approach relies on a specific rank condition on the collected data, guaranteeing persistent excitation of the unknown system, which allows for direct ASF and controller design using finite-length data, bypassing the need for explicit system identification. The effectiveness of our proposed approach is demonstrated via a case study. Future work will focus on extending our framework to construct ASFs for a broader class of nonlinear systems beyond polynomial dynamics.

\bibliographystyle{alpha}
\bibliography{biblio}

\end{document}